\documentclass{easychair}
\usepackage{lmodern}
\usepackage[utf8]{inputenc}
\usepackage{latexsym}
\usepackage[T1]{fontenc}
\usepackage{amssymb}
\usepackage{fix-cm} 
\usepackage{stmaryrd}
\usepackage{empheq}
\usepackage{fancyhdr}
\usepackage{geometry}
\usepackage{amsmath}
\usepackage{footmisc}
\usepackage{eulervm}

\usepackage{graphicx}
\usepackage{xcolor}
\usepackage{tikz}
\usepackage{tabularx}
\usepackage{array}
\usepackage{eqparbox}
\usepackage{booktabs}

\usepackage{url}
\usepackage{hyperref}

\newcommand{\pc}{\mathbf{P}}
\newcommand{\oc}{\mathbf{O}}

\newcommand{\pp}{\mathbb{P}}

\newcommand{\cc}{\mathbb{C}}

\newtheorem{theorem}{Theorem}
\newtheorem{definition}{Definition}
\newtheorem{proposition}{Proposition}
\newtheorem{example}{Example}

\title{Rough Contact in General Rough Mereology}
\author{\textbf{A. Mani}}
\institute{International Rough Set Society\\
9/1B, Jatin Bagchi Road\\
Kolkata(Calcutta)-700029, India\\
Email: \email{$a.mani.cms@gmail.com$}\\
Homepage: \url{http://www.logicamani.in}}

\authorrunning{A Mani}

\titlerunning{General Rough Mereology}

\begin{document}

\maketitle

\begin{abstract}
Theories of rough mereology have originated from diverse semantic considerations from contexts relating to study of databases, to human reasoning. These ideas of origin, especially in the latter context, are intensely complex. In this research, concepts of rough contact relations are introduced and rough mereologies are situated in relation to general spatial mereology by the present author. These considerations are restricted to her rough mereologies that seek to avoid contamination.

\end{abstract}
\section{Introduction}
\label{sect:introduction}
In a semantic domain, an object is \emph{crisp} if it can be discerned perfectly. Rough set theory (RST) is a formal approach to vagueness (typically concerning non-crisp or rough objects) that is studied in multiple semantic domains -- the motive can vary from insights into databases to human reasoning. Naturally this results in a number of not necessarily equivalent semantic approaches. They \cite{am240,am501,am3930,ppm2,yy2015} correspond to distinct methods of reasoning about rough objects and related approximation operations, and possibly other types of objects (and operations). In general, these reasoning methods are not necessarily visual or spatial in any sense. A part of such reasoning has connections with popular modal logic because the lower approximation $x^l$ of an object $x$ is interpretible as the definite part of $x$ and the upper approximation $x^u$ as a \emph{possible definition} of it.

Mereological aspects in the context of modeling human reasoning are relatively more complex than those relating to databases, and fault tolerant systems. If wild numeric assumptions are avoided, then rough reasoning about databases can be intuitive and rational. This fragment is of natural interest for transforming models of human reasoning couched in ideas of approximation, contradiction, inconsistency, and relative atomicity. 

A general overview of mereology can be found in \cite{rgac15}. Predicates like \emph{is part of} ($\pc$), \emph{overlaps} ($\oc$), \emph{is in contact with} ($\cc$) and others may be of a basic or a derived nature (that is defined in terms of basic predicates) are considered in such studies. Classical extensional mereology (CEM) proceeds from a single primitive concept of transitive parthood, and Kuratowski general extensional mereology (KGEMT, as defined in \cite{acav2019}) takes transitive $\pc$ and $\cc$ as primitives.  

At least two approaches to rough mereology \cite{am240,am3930,am9969,lp2019} are known in the literature on rough sets. They differ substantially from the mereologies mentioned earlier. The mereology of \cite{ps} is based on the degree of inclusion of an object in another (as in '\emph{ $a$ is included in $b$ to the degree $r$}')\cite{ag2009}, and consequently the basic parthood predicate is ternary and non-transitive for fixed values of the degree.  In \cite{lp2019}, spatial representation aspects of the membership degree/mass based rough mereology (MDRM) are considered. Interestingly the role of contact-like predicates is understood from a CEM perspective. The concerns of rough mereologies due to the present author in \cite{am240,am3930,am9969} (collectively referred to as RMCA: Rough Mereology with Contamination Avoidance) are founded in ideas of operators that can approximate, granularity, contamination avoidance, and rough dependence. It is possible to do RST from an abstract operator theoretic perspective without reference to granules (see \cite{gcd2018}) -- related mereology form a subset of the present approach. Negations and ortho pairs play a dominant role in the theory, and guarantee improved semantics. 

\emph{In this research, the nature of mereological perspectives in general RST is clarified and the problem of defining minimal frameworks for RMCA and rough contact relations (for a spatial mereology) are investigated}.

\section{Background and Terminology}

A \emph{general approximation space} is a relational system of the form $S\,=\,\left\langle \underline{S},\,R\right\rangle $  with $\underline{S}$ being a set (in ZFC) and $R$ a binary relation on it. \textsf{The symbol $S$ will be used for $\underline{S}$ if such usage is clear from the context}. By \emph{relation based RST} (\textsf{RBRST}), will be meant any RST over such general approximation spaces. If $R$ is replaced by a covering $\mathcal{C}$ of $\underline{S}$, then the pair is said to be a \emph{covering approximation space} (CAS).

In classical rough sets, $R$ is an equivalence. On the power set $\wp (S)$, lower and upper
approximations of a subset $A\in \wp (S)$ operators, apart from the usual Boolean operations, are defined as per: 
\[A^l = \bigcup_{[x]\subseteq A} [x] \; ; \; A^{u} = \bigcup_{[x]\cap A\neq \varnothing } [x],\,\]
with $[x]$ being the equivalence class generated by $x\in S$. If $A, B\in \wp (S)$, then $A$ is said to be \emph{roughly included} in $B$ $(A\sqsubseteq B)$ if and only if $A^l \subseteq B^l$ and $A^u\subseteq B^u$. $A$ is roughly equal to $B$ ($A\approx B$) if and only if $A\sqsubseteq B$ and $B\sqsubseteq A$. The positive, negative and boundary region determined by a subset $A$ are respectively $A^l$, $A^{uc}$ and $A^{u}\setminus A^l$ respectively. $\sqsubseteq$ is an example of a transitive and reflexive parthood relation in $\wp(S)$.

Boolean algebra with approximation operators constitutes a semantics for classical RST and RBRST (though not satisfactory). More generally it is possible to replace $\wp (S)$ by some set with a part-hood relation and some approximation operators defined on it \cite{am240}. The associated semantic domain in the sense of a collection of restrictions on possible objects, predicates, constants, functions and low level operations on those is referred to as the classical semantic domain for general RST. In contrast, sets of roughly equivalent or relatively indiscernible objects can be associated with
 a \emph{rough semantic domain}. In the literature \cite{am501}, many others including hybrid semantic domains have been used .

Data analysis maybe intrusive (invasive) or non-intrusive relative to the assumptions made on the dataset used in question \cite{gdu}. Non-invasive data analysis was defined in a vague way in \cite{gdu} as one that 

Is based on the idea of \emph{data speaking for themselves},

Uses minimal model assumptions by drawing all parameters from the observed data, and

Admits ignorance when no conclusion can be drawn from the data at hand.

Key procedures that have been deemed to be non-invasive in \cite{gdu} include data discretization (or horizontal compression), randomization procedures, reducts of various kinds within rough set data analysis, and rule discovery with the aid of maximum entropy principles. 

The concept of \emph{contamination} was introduced in \cite{am99} and explored in \cite{am240,am501,am9114} by the present author. It maybe viewed as a distinct minimalist approach that takes the semantic domains involved into account and has the potential to encompass the three principles of non-intrusive analysis. Some sources of contamination are those about distribution of variables, introduction of assumptions valid in one semantic domain into another by oversight, numeric functions used in rough sets (and soft computing in general) and fuzzy representation of linguistic hedges. The contamination problem in simplified terms is that of reducing artificial constructs in \textsf{RST} towards capturing essential rough reasoning at that level. Reduction of contamination is relevant in all model/algorithm building contexts of formal approaches to vagueness. 

Granules or information granules are often the minimal discernible concepts that can be used to construct all relatively crisp complex concepts in a vague reasoning context. Such constructions typically depend on a substantial amount of assumptions made by the theoretical approach employed \cite{am240,am501,am9411,tyl}. In the present author's axiomatic approach to granularity \cite{am240,am9114,am9411,am501,am3930}, fundamental ideas of non-intrusive data analysis have been critically examined and methods for reducing contamination of data (through external assumptions) have been proposed. The need to avoid over-simplistic constructs like rough inclusion functions have been stressed in the approach by her. New granular measures that are compatible with rough domains of reasoning, and granular correspondences that avoid measures have also been invented in the papers.

\subsection{Contact Relations}
A \emph{precontact relation} $\cc $ over a lattice $L$ is a binary relation that satisfies C1, C6, and C7, while a contact relation over a bounded distributive lattice is one that satisfies C1, C2, C3, C4 and C5. These are of fundamental relevance in spatial mereology.
\begin{align*}
\cc ab \longrightarrow 0 < a \, \&\, 0 < b \tag{C1}\\
\cc a (b\vee e) \leftrightarrow \cc ab \text{ or } \cc ae   \tag{C6}\\
\cc (b\vee e)a \leftrightarrow \cc ba \text{ or } \cc ea   \tag{C7}\\
\cc ab \longrightarrow \cc ba \tag{C2}\\
\cc ab \,\& \, b\leq e \longrightarrow \cc ae  \tag{C3}\\
\cc a (b\vee e) \longrightarrow \cc ab \text{ or } \cc ae \tag{C4}\\
0 < a\wedge b \longrightarrow \cc ab    \tag{C5}
\end{align*}

A \emph{partial algebra} $P$ is a tuple of the form $\left\langle\underline{P},\,f_{1},\,f_{2},\,\ldots ,\, f_{n}, (r_{1},\,\ldots ,\,r_{n} )\right\rangle$ with $\underline{P}$ being a set, $f_{i}$'s being partial function symbols of arity $r_{i}$. The interpretation of $f_{i}$ on the set $\underline{P}$ should be denoted by $f_{i}^{\underline{P}}$, but the superscript will be dropped in this paper as the application contexts are simple enough. If predicate symbols enter into the signature, then $P$ is termed a \emph{partial algebraic system}.   

In a partial algebra, for two terms $s,\,t$, $s\,\stackrel{\omega}{=}\,t$ shall mean, if both sides are defined then the two terms are equal (the quantification is implicit). $s\,\stackrel{\omega ^*}{=}\,t$ shall mean if either side is defined, then the other is and the two sides are equal (the quantification is implicit).

\section{Variants of Granular Operator Spaces}

Granular operator spaces and related variants can be directly constructed from records of human reasoning, databases or from partial semantics of general rough sets. They are mathematically accessible powerful abstractions for handling semantic questions, formulation of semantics and the inverse problem. As many as six variants of such spaces have been defined by the present author - all these can be viewed as special cases of a set theoretic and a relation-theoretic abstraction with abstract operations from a category-theory perspective. Strictly speaking, they are partial algebraic systems, the '\emph{space}' is because of mathematical usage norms. 

In a \emph{high general granular operator space}, introduced below, aggregation/co-aggregation operations ($\vee, \,\wedge$) are conceptually separated from the binary parthood relation ($\pc$), and a basic partial order ($\leq$). In real-life information processing, it often happens that many instances of aggregations (disjunctions), co-aggregation (conjunctions) and implications are ignored because of laziness or incompatibility -- this justifies the use of partial operations.

\begin{definition}\label{gfsg}

A \emph{High General Granular Operator Space} (\textsf{GGS}) $\mathbb{S}$ shall be a partial algebraic system  of the form $\mathbb{S} \, =\, \left\langle \underline{\mathbb{S}}, \gamma, l , u, \pc, \leq , \vee,  \wedge, \bot, \top \right\rangle$ with $\underline{\mathbb{S}}$ being a set, $\gamma$ being a unary predicate that determines $\mathcal{G}$ (by the condition $\gamma x$ if and only if $x\in \mathcal{G}$) 
an \emph{admissible granulation}(defined below) for $\mathbb{S}$ and $l, u$ being operators $:\underline{\mathbb{S}}\longmapsto \underline{\mathbb{S}}$ satisfying the following ($\underline{\mathbb{S}}$ is replaced with $\mathbb{S}$ if clear from the context. $\vee$ and $\wedge$ are idempotent partial operations and $\pc$ is a binary predicate. Further $\gamma x$ will be replaced by $x \in \mathcal{G}$ for convenience.):

\begin{align*}
(\forall a, b) a\vee b \stackrel{w}{=} b\vee a \;;\; (\forall a, b) a\wedge b \stackrel{w}{=} b\wedge a \\
(\forall a, b) (a\vee b) \wedge a \stackrel{w}{=} a \; ;\; (\forall a, b) (a\wedge b) \vee a \stackrel{w}{=} a\\
(\forall a, b, c) (a\wedge b) \vee c \stackrel{w}{=} (a\vee c) \wedge (b\vee c)\\
(\forall a, b, c) (a\vee b) \wedge c \stackrel{w}{=} (a\wedge c) \vee  (b\wedge c)\\
(\forall a, b) (a\leq b \leftrightarrow a\vee b = b \,\leftrightarrow\, a\wedge b = a  )\\
(\forall a \in \mathbb{S})\,  \pc a^l  a\,\&\,a^{ll}\, =\,a^l \,\&\, \pc a^{u}  a^{uu}  \\
(\forall a, b \in \mathbb{S}) (\pc a b \longrightarrow \pc a^l b^l \,\&\,\pc a^u  b^u)\\
\bot^l\, =\, \bot \,\&\, \bot^u\, =\, \bot \,\&\, \pc \top^{l} \top \,\&\,  \pc \top^{u} \top \\
(\forall a \in \mathbb{S})\, \pc \bot a \,\&\, \pc a \top
\end{align*}

Let $\pp$ stand for proper parthood, defined via $\pp ab$ if and only if $\pc ab \,\&\,\neg \pc ba$). A granulation is said to be admissible if there exists a term operation $t$ formed from the weak lattice operations such that the following three conditions hold:
\begin{align*}
(\forall x \exists
x_{1},\ldots x_{r}\in \mathcal{G})\, t(x_{1},\,x_{2}, \ldots \,x_{r})=x^{l} \\
\tag{Weak RA, WRA} \mathrm{and}\: (\forall x)\,(\exists
x_{1},\,\ldots\,x_{r}\in \mathcal{G})\,t(x_{1},\,x_{2}, \ldots \,x_{r}) =
x^{u},\\
\tag{Lower Stability, LS}{(\forall a \in
\mathcal{G})(\forall {x\in \underline{\mathbb{S}}) })\, ( \pc ax\,\longrightarrow\, \pc ax^{l}),}\\
\tag{Full Underlap, FU}{(\forall
x,\,a \in\mathcal{G})(\exists
z\in \underline{\mathbb{S}}) )\, \pp xz,\,\&\,\pp az\,\&\,z^{l}\, =\,z^{u}\, =\,z,}
\end{align*}
\end{definition}
\emph{The conditions defining admissible granulations mean that every approximation is somehow representable by granules in a set theoretic way, that every granule coincides with its lower approximation (granules are lower definite), and that all pairs of distinct granules are contained in definite objects (those that coincide with their own lower and upper approximations).} Special cases of the above are defined next.

\begin{definition}
\begin{itemize}
\item {In a \textsf{GGS}, if the parthood is defined by $\pc ab$ if and only if $a \leq b$ then the \textsf{GGS} is said to be a \emph{high granular operator space} \textsf{GS}.}
\item {A \emph{higher granular operator space} (\textsf{HGOS}) $\mathbb{S}$ is a \textsf{GS} in which the lattice operations are total.}
\item {In a higher granular operator space, if the lattice operations are set theoretic union and intersection, then the \textsf{HGOS} will be said to be a \emph{set HGOS}. }
\end{itemize}
\end{definition}

\begin{example}
 
A set HGOS is intended to capture contexts where all objects are described by sets of attributes with related valuations (that is their properties). So objects can be associated with sets of properties (including labels possibly). A more explicit terminology for the concept, may be \emph{power set derived HGOS}(that captures the intent that subsets of the set of all properties are under consideration here). Admittedly, the construction or specification of such a power set is not necessary. In a HGOS, such set of sets of properties need not be the starting point.
 
The difference between a HGOS and a set HGOS at the practical level can be interpreted at different levels of complexity. Suppose that the properties associated with a familiar object like a cast iron frying pan are known to a person $X$, then it is possible to associate a set of properties with valuations that are sufficient to define it. If all objects in the context are definable to a \emph{sufficient level}, then it would be possible for $X$ to associate a set HGOS (provided the required aspects of approximation and order are specifiable). 

It may not be possible to associate a set of properties with the same frying pan in a number of scenarios. For example, another person may simply be able to assign a label to it, and be unsure about its composition or purpose. Still the person may be able to indicate that another fryng pan is an approximation of the original frying pan. In this situation, it is more appropriate to regard the labeled frying pan as an element of a HGOS. 

A nominalist position together with a collectivization property can also lead to HGOS that is not a set HGOS.
\end{example}

\begin{definition}
An element $x\in\mathbb{S}$ is said to be \emph{lower definite} (resp. \emph{upper definite}) if and only if $x^l\, =\,x$ (resp. $x^u\, =\,x$) and \emph{definite}, when it is both lower and upper definite. $x\in \mathbb{S}$ is also said to be \emph{weakly upper definite} (resp \emph{weakly definite}) if and only if $ x^u\, =\,x^{uu} $ (resp $ x^u\, =\,x^{uu} \,\&\, x^l =x$ ). Any one of these five concepts may be chosen as a concept of \emph{crispness}. 
\end{definition}

In granular operator spaces and generalizations thereof, it is possibly easier to express singletons and the concept of rough membership functions can be generalized to these from a granular perspective. For details see  \cite{am501,am9114}. Every granular operator space can be transformed to a higher granular operator space, but to speak of this in a rigorous way, it is necessary to define related morphisms and categories\cite{am501}.
 
\begin{proposition}
Every HGOS is a GS, and every GS is a GGS
\end{proposition}

\paragraph{Rough Objects}
A rough object cannot be known exactly in a rough semantic domain, but can be represented in a number of ways.  The following representations of \emph{rough objects} have been either considered in the literature (see \cite{am501}) or are reasonable concepts that work in the absence of a negation-like operation or relation: {\emph{any pair of definite elements} of the form $(a , b)$ satisfying $a < b $}, \emph{any distinct pair of elements} of the form $(x^l ,x^u)$, \emph{intervals of the form} $(x^l, x^u)$, and \emph{interval of the form} $(a, b)$ satisfying $a\leq b$ with $a, b$ being definite elements.

\section{Extended Example, Fusion}

The difference between fusion ($\mathfrak{F}\subseteq S \times \wp (S)$) and sum ($\Sigma \subseteq S \times \wp (S)$) predicates is relevant in RMCA. Avoiding issues relating to existence, the predicates can be defined as  

\begin{align}
\Sigma a B \stackrel{\vartriangle}{\leftrightarrow} B\subseteq \pc(a) \subseteq \bigcup \{\oc (x): \, x\in B\} \tag{msum}\\
\mathfrak{F} a B \stackrel{\vartriangle}{\leftrightarrow} \oc (a) = \bigcup \{\oc (x): \, x\in B\}   \tag{fusion}
\end{align}

For a set $S$ endowed with a binaryparthood relation $\pc$, the set of upper and lower bounds of a subset $X$ are defined by 
\begin{align*}
UB(X) = \{a:\, (\forall x\in X) \pc xa\}   \tag{Upper Bounds}\\
LB(X) = \{a:\, (\forall x\in X) \pc ax\}   \tag{Lower Bounds}
\end{align*}
$S$ is said to be \emph{separative} if and only if SSP (strong supplementation) holds.
\[(\forall a b )(\neg \pc ab \longrightarrow (\exists z) (\pc za \,\&\, \neg \pc zb \,\&\, \neg \pc bz))   \tag{SSP}\]
\begin{theorem}[\cite{rafal2013}]
All of the following hold:
\begin{itemize}
\item {If $\pc$ is reflexive, then  a fusion of $B$ is a mereological sum if it is an upper bound of $B$: \[(\forall a\in S)(\forall B\in \wp (S))(B\subseteq \pc (a) \,\&\, \mathfrak{F} a B \longrightarrow \Sigma a B)\]}
\item {If $\pc$ is transitive and separative then every sum is a fusion and conversely.}
\item {If $\pc$ is transitive and separative then every binary fusion is a binary sum}
\end{itemize}
\end{theorem}

\begin{example}{Fusion and Decisions}\label{progress}
Let $S= \{a,\, b,\, c,\, e, \, f \}$ be a set with parthood $\pc$ defined as the reflexive completion of  
\[ \{(a, c)\, (b, c), \, (a, e),\, (b, e) \}. \]
 
 If $K= \{a,\, b,\, c,\, e\}$, then $\mathfrak{F}cK$ and $\mathfrak{F}eK$ hold. But, $UB(K) = \emptyset$.
 
 Suppose $S$ represents the respective diagnosis of five doctor teams $X$, $W$, $Z$, $E$, and $F$,
 on the basis of diagnostic information indicated in the decision table below. Consider columns $1$, $4$ and $6$ alone first. The sixth column indicates the team type (based on the best performing doctor in the team) involved in the diagnosis. Assume that the doctors are essentially lower approximating an \emph{ideal diagnosis} and that $\pc  \beta \alpha$ means '$\alpha$ is a better diagnosis than $\beta$'.

\begin{table}[hbt]
 \centering
\begin{tabular}{llllllll}
\toprule
\textsf{Doctors} &\textsf{Attribute:1--3} &\textsf{Attributes:4--6} &\textsf{Attributes: 7--9} &\textsf{Diagnosis} & \textsf{Remark} & l & u   \\
\midrule
$X$  &smm & www & nnw & $a$ & General & $X$ & $Z$\\
\midrule
$W$  &mww & swm  &nnn &  $b$ & General & $W$ & $E$ \\
\midrule
$Z$  &smm &mwm  &wmw &   $c$ & Specialist & $Z$ & $Z$\\
\midrule
$E$  &msw  &swm & mms & $e$ & Specialist & $W$ & $E$\\
\midrule
$F$  &mss & mwm &mws &   $f$ & Specialist & $F$ & $F$\\
\end{tabular}
\caption{Doctors and Diagnosis}\label{doci}
\end{table} 
 
 Mereological fusion in the context corresponds to combining expert information. It cannot be used in the context to arrive at any all encompassing ideal diagnosis.

The attributes used for the diagnosis are encoded as per: s- severe, m-moderate, w-weak, n-not available. Thus the string \textsf{smm} in the second cell is intended to mean that the valuation for attribute 1 is \textsf{s}, attribute 2 is \textsf{m} and attribute 3 is \textsf{m}. Further suppose that the attributes are potentially causally related, and that the valuations assigned by the doctors are dependent on their own perspectives. The lower and upper approximations of the teams relative to their potential in the context is indicated in the last two columns. 

It is not hard to obtain a granulation based on a simple ordering of the attribute valuation. In practice, the situation is usually more complex. A partial order on $S$ can be specified based on the training of doctor teams and a GGS can be defined on the basis of this information on $S$. It should also be easy to see that the rough inclusion function perspective in the context does not correspond to the approximations. 
\end{example}

\section{Rough Sets and Contact Relations}

In RMCA, the primitives are essentially determined by the fundamental predicates and operations of the granular operater space variant used. For a reasonable comparison with frameworks like KGEMT, it makes sense to restrict considerations to HGOS.

\emph{Two elements of a granular operator space are in rough contact if and only if they have things that share a granule in a perspective or a process} -- this loose statement is very general and can be used to define concepts of rough contact relations in the rough semantic domain, the classical semantic domain and other hybrid semantic domains. This means that a large number of such concepts are possible. In the absence of a granulation, two objects in an abstract approximation space may be said to be in \emph{rough contact} if and only if they share a definite object in some sense. 

If two objects share some ontological dependence then also they need not be in contact in a physical sense, but can be said to be in contact in a visual sense. Some work on the latter is known in the context of proximity based approaches in rough sets \cite{bid,cgjp2016}. These concepts are not necessarily granular and require much stronger conditions. 

Concepts of rough dependence have been defined and used in models and for comparison by the present author in \cite{am3930,am9501,am9411}. If the rough dependence of $a$ on $b$ is $c$, then in some sense, $a$ can be said to be in contact with $b$ - but this requires additional conditions to be satisfied by the context as modeled by the GGS. These will appear in a separate paper.

A number of new rough contact relations are introduced next.

\begin{definition}
In a GGS $\mathbb{S} \, =\, \left\langle \underline{\mathbb{S}}, \mathcal{G}, l , u, \pc, \leq , \vee,  \wedge, \bot, \top \right\rangle,$ the relation $\Re_\alpha$  is intended to indicate a \emph{rough contact relation} of type $\alpha$ as per the following definitions:
\begin{align*}
\Re_a xb \text{ iff } (\exists e , f\in S)(\exists g\in \mathcal{G}) \pc ex \,\&\, \pc f b\,\&\, \pc eg \, \&\, \pc fg    \tag{Type a}\\ 
\Re_o ab \text{ if and only if } (\exists e\in \mathcal{G})\, \pc ea \, \&\, \pc eb \tag{Type o} \\
\Re_1 ab \text{ if and only if } (\exists f) \, \pc fa \,\&\, f^u \leq b^u\, \&\, b^u \leq f^u \tag{Type 1}\\
\Re_2 ab  \text{ if and only if } (\exists e, f \in \mathcal{G} ) \, \pc ea^l \,\&\, \pc eb^u \,\&\, \pc f a^u \,\&\, \pc f b^l   \tag{Type 2} \\
\Re_3 ab  \text{ if and only if } (\exists e\in \mathcal{G}) \, \pc ea^u \,\&\, \pc e b^u \tag{Type 3} 
\end{align*}
\end{definition}

The meaning of these rough contact relations is as follows:
\begin{itemize}
\item {Two objects $x$ and $b$ are in type-a contact if and only if parts of $x$ and $b$ are in contact within a granule. This is obviously not necessarily a contamination-free concept in general because parts of granules may require knowledge about all kinds of objects.}
\item {$a$ and $b$ are in type-o contact if and only if they share a granule. This can for example, be about sharing a basic kind of 'redness'.}
\item {$a$ and $b$ are in type-1 contact if and only if there is a part of $a$ whose upper approximation coincides (relative to $\leq$) with that of $b$.}
\item {$a$ and $b$ are in type-2 contact if and only if there exists a granule $e$ that is definitely in $a$ and possibly in $b$, and a granule $f$ that is definitely in $b$ and possibly in $a$. This is just one of the possible readings of the definition - it may be wrong.}
\item {$a$ and $b$ are in type-3 contact if and only if there exists a granule $e$ that is part of the upper approximations of $a$ and $b$ respectively. If the parthood is set inclusion, then type-3 contact is a form of contact relation. This can be related to sharing a relatively basic property in a reasonable speculative perspective.}
\end{itemize}
In one perspective, these are generalizations of the idea of descriptive proximity \cite{cgjp2016} that rely on granulations instead of the restrictive probe functions used therein. The connection is illustrated by Example \ref{proxim} (adapted from \cite{am501} by the present author):

\begin{example}\label{proxim}
 
In the case study on numeric visual data including micro-fossils with the help of nearness and remoteness granules in \cite{cgjp2016}, the difference between granules and approximations is very fluid as the precision level of the former can be varied. The data set consists of values of probe functions that extract pixel data from images of micro-fossils trapped inside other media like amethyst crystals. The same strategy can be used in a number of similar visual data sets including the analysis of paintings.

The idea of remoteness granules is relative to a fixed set of nearness granules formed from upper approximations - so the approach is about reasoning with sets of objects which in turn arise from tolerance relations on a set. In \cite{am501,am9114}, the computations are extended to form maximal antichains at different levels of precision towards working out the best antichains from the point of view of classification.

Let $X$ be an object space consisting of representation of some patterns and $\Phi : X\longmapsto \mathbb{R}^n$ be a \emph{probe function}\index{Function!Probe}, defined by 
$\Phi (x)\, =\,(\phi_1 (x), \phi_2 (x), \ldots , \phi_n (x)),$
where $\phi_i (x)$ is intended as a measurement of the $i$th component in the feature space $\Im (\Phi)$. The concept of descriptive intersection of sets permits migration from classical ideas of proximity to ones based on descriptions. A subset $A\subseteq X$'s descriptive intersection with subset $B\subseteq X$ is defined by $A\cap_\Phi B\, =\,\{x\in A\cup B : \Phi(x) \in \Phi(A) \,\&\,\Phi (x) \in \Phi (B) \}$.
$A$ is then \emph{descriptively near}\index{Descriptively Near} $B$ if and only if their descriptive intersection is nonempty. Peter's version of proximity $\pi_\Phi$ is defined by $ A \pi_\Phi B \leftrightarrow \Phi(A)\cap \Phi(B) \neq \emptyset $. In \cite{cgjp2016}, weaker implications for defining \emph{descriptive nearness} are considered : 
$A\cap_\Phi B\neq \emptyset \rightarrow A\delta_\Phi B$. Specifically, if $\delta$ is a proximity on $R^n$, then a descriptive proximity $\delta_{\Phi}$ is definable via $A\delta_\Phi B \leftrightarrow \Phi(A)\delta \Phi (B)$.

 In these contexts, the concepts of rough contact can be used for greater flexibility in application. Specifically when the intended interpretation of a picture has to do with complex geometry, and pixel values do not suffice, then it would be better to use rough contact relations for knowledge representation. 
 
\end{example}

\begin{definition}
A higher general granular operator space $S$ will be said to be \emph{tractable} if and only if 
$(\forall x\in S)(\exists b\in \mathcal{G} )\, \pc ab \vee \pc b a$ ({Tractability}).
\end{definition}

\begin{theorem} All of the following hold:

$\Re_a$ is a symmetric relation that is not reflexive or  extensional or transitive in general.

If $S$ is tractable, then $\Re_a$ is reflexive.

$\Re_o,\, \Re_2, $ are symmetric relations that are not reflexive or  extensional or transitive in general.

$\Re_1$ is not symmetric or transitive, but is reflexive.

$\Re_3$ is reflexive if $S$ has no zero element or if the bottom element is a granule.
\end{theorem}
\begin{proof}
Each of the defining conditions (of $\Re_i a b$ for $i\in \{a, o, 2, 3\}$) is of the form \[(\exists e f)(\exists h \in \mathcal{G}) \Phi(x, b, e, f, h) \,\&\, \Phi(b, x, e, f, h)\]  or $(\exists e f) \Phi(x, b, e, f) \,\&\, \Phi(b, x, e, f)$ for some formula function $\Phi$ and with some of the variables being possibly identical or missing. Therefore $\Re_i$ is symmetric. Reflexivity fails for $\Re_a$ because of empty sets.

$\Re_1$ fails to be symmetric in the context of classical rough sets. Let $ b\subset b^u = g_1\cup g_2$ for some distinct non singleton granules $g_1, g_2$, and $a^u = b^u \cup g_3$ for a distinct non singleton granule $g_3$. If it is also the case that $a \subset a^u$, then it is easy to find a suitable $a$ for which $\Re_1 ab$ and $\neg \Re_1 ba$ holds.
\end{proof}

The way in which $\Re_2$ fails to be reflexive motivates the following definition

\begin{definition}
By a \emph{granule aware} element $x$ of a higher granular operator space $S$ will be one that satisfies 
$(\exists a \in \mathcal{G})\, \pc a x$. The set of such elements will be denoted by $S_o$ and any reflexive binary relation on $S_o$ will be said to be a  \emph{granule aware} relation.
\end{definition}

\begin{proposition}
$\Re_2$ and $\Re_o$ are granule aware relations.
\end{proposition}

Based on each of these rough contact relations, additional relations like \emph{roughly-disconnected, externally roughly connected, and tangential proper rough part} can be defined as is done in case of contact relations. \emph{These concepts have a clear descriptive perspective in the contexts of Example \ref{proxim}. Further, they permit a correspondence with spatial mereology based on KGEMT}.

The defined rough contact relations have a number of properties in the context of CAS:  
\begin{theorem}\label{cas1}
 Let $S\,=\, \left\langle\underline{S}, \mathcal{S} \right\rangle$ be a CAS with $\bigcup \mathcal{S} = \underline{S}$.  If $l$ and $u$ are any lower and upper approximation that satisfy the same conditions as that of granular operator spaces, $\subseteq$ being the parthood relation, and if $\mathcal{S}$ is an admissible granulation, then

$\Re_a$ is a contact relation.

$\Re_o$ satisfies C1, C2, and  C3 but C4 and C5 may not hold, while $\Re_1$ does not satisfy C1, C2, C4 and C5 in general.

$\Re_2$ satisfies C1, C2 and C3 but C4 and C5 do not hold in general, while $\Re_3$ satisfies C1, C2 and C3.
\end{theorem}
\begin{proof}
The definition of $\Re_a$ is very mild. Because every subset of the set $S$ includes subsets that must necessarily be contained in a granule, C1, C4 and C5 follow. This happens because the granulation is a proper cover.

$\Re_o$ does not satisfy C4 even in the context of classical rough sets. To see this, let $\mathcal{S}$ be a cover including the granulation (a partition), $g = \{1, 2, 3\}$ be the only granule contained in the set $a$, $b\cap g = \{1, 2\}$, and $c\cap g = \{3\}$. Then $\Re_oa (b\cup c)$ holds, but both $\Re_o ab$ and $\Re_o ac$ do not hold.
In this situation $\emptyset \subset a\cap b$, but $\Re_o ab$ is not valid in the model. So C5 is also not generally true.

\end{proof}

\paragraph{Conclusion} In this research, concepts of rough contact relations have been introduced in an accessible way in the context of a very general approach to rough sets by the present author. A comprehensive and specialized  version of this research will appear separately.   

\bibliographystyle{plain}
\bibliography{algroughffff}
\appendix
\section{Proof of Theorem 2}

\begin{proof}
\begin{itemize}
\item {Each of the defining conditions (of $\Re_i a b$ for $i\in \{a, o, 2, 3\}$) is of the form \[(\exists e f)(\exists h \in \mathcal{G}) \Phi(x, b, e, f, h) \,\&\, \Phi(b, x, e, f, h)\]  or $(\exists e f) \Phi(x, b, e, f) \,\&\, \Phi(b, x, e, f)$ for some formula function $\Phi$ and with some of the variables being possibly identical or missing. Therefore $\Re_i$ is symmetric. Reflexivity fails for $\Re_a$ because of empty sets, while transitivity fails for better reasons.}
\item {$\Re_1$ fails to be symmetric in the context of classical rough sets. Let $ b\subset b^u = g_1\cup g_2$ for some distinct non singleton granules $g_1, g_2$, and $a^u = b^u \cup g_3$ for a distinct non singleton granule $g_3$. If it is also the case that $a \subset a^u$, then it is easy to find a suitable $a$ for which $\Re_1 ab$ and $\neg \Re_1 ba$ holds.}
\end{itemize}

\end{proof}
\section{Proof of Theorem 3}
\begin{proof}
\begin{itemize}
\item {The definition of $\Re_a$ is very mild. Because every subset of the set $S$ includes subsets that must necessarily be contained in a granule, C1, C4 and C5 follow. This happens because the granulation is a proper cover.}
\item {$\Re_o$ does not satisfy C4 even in the context of classical rough sets. To see this,
\begin{itemize}
\item {let $\mathcal{S}$ be a cover including the granulation (a partition) and sets $a, b, c$,}
\item {$g = \{1, 2, 3\}$ be the only granule contained in the set $a$,}
\item {$b\cap g = \{1, 2\}$, and}
\item {$c\cap g = \{3\}$.}
\item {Then $\Re_oa (b\cup c)$ holds, but both $\Re_o ab$ and $\Re_o ac$ do not hold.}
\item {In this situation $\emptyset \subset a\cap b$, but $\Re_o ab$ is not valid in the model. So C5 is also not generally true.}
\end{itemize}}
\item {For $\Re_1$, the required counterexamples can again be generated in the classical context.}
\item {$\Re_2 ab$ implies that the lower approximations of $a$ and $b$ contain granules. So $a$ and $b$ must be nonempty. That is C1 is satisfied by $\Re_2$. C2 and C3 are easy to verify. The counterexamples for C4 and C5 can be generated in the classical context.}
\item {For $\Re_3$ the verification is direct.}
\end{itemize}
\end{proof}
\end{document}